\documentclass[journal]{IEEEtran}      

\IEEEoverridecommandlockouts                              

\usepackage{amssymb, amsmath, latexsym, amsfonts, mathrsfs, amsbsy,mathtools}
\usepackage{amsthm}
\usepackage[english]{babel}
\usepackage[dvips]{graphicx}
\graphicspath{./Images/}
\usepackage[utf8]{inputenc}
\usepackage{epsfig, color}
\usepackage{cite}
\usepackage{listings}
\usepackage{subfloat}
\usepackage{makeidx, glossaries}
\usepackage{url}
\usepackage[caption=false, font=footnotesize]{subfig}
\usepackage{mathbbol}
\usepackage{stmaryrd}
\usepackage{enumerate}
\usepackage{epigraph}
\usepackage{multirow}
\usepackage{footnote}
\usepackage{bbm}

\usepackage[T1]{fontenc}

\usepackage[algo2e]{algorithm2e} 
\usepackage{lipsum}

\newtheorem{theorem}{Theorem}
\newtheorem{lemma}{Lemma}

\newtheorem{assumption}{Assumption}

\newtheorem{remark}{Remark}

\DeclareMathOperator{\diag}{diag}

\allowdisplaybreaks

\newcommand{\black}[1]{\color{black}{#1}\color{black}\phantom{}}

\usepackage{physics}
\usepackage{color}
\usepackage{array,booktabs,calc}
\usepackage{float}
\usepackage{graphics} 
\usepackage{epsfig} 
\usepackage{times} 
\usepackage{multirow}
\usepackage{amsmath}
\usepackage{amsbsy}
\usepackage{supertabular}
\usepackage{tkz-graph}
\usepackage{etex}
\usepackage{pgf, tikz}
\usepackage{bm}
\usetikzlibrary{shapes,shapes.geometric,arrows,fit,calc,positioning,automata,through,intersections}
\tikzset{diamond state/.style={draw,diamond}}

\title{\LARGE 
Distributed estimation and control of node centrality\\
in undirected asymmetric networks
}

\author{Eduardo~Montijano, Gabriele Oliva and Andrea Gasparri 
\thanks{E.~Montijano is with Instituto de Investigaci\'{o}n en Ingenier\'{\i}a de Arag\'{o}n (I3A), Universidad de Zaragoza, Zaragoza, Spain.
\mbox{E-mail: \textit {emonti@unizar.es}}}%
\thanks{G.~Oliva is with the Unit of Automatic Control, Department of Engineering, Universit\`a Campus Bio-Medico di Roma, Rome, Italy.
\mbox{E-mail: \textit {g.oliva@unicampus.it}}}%
\thanks{A.~Gasparri is with the Department of Computer Science and Automation,
       University of ``Roma Tre'', Rome, Italy
\mbox{E-mail: \textit {gasparri@dia.uniroma3.it}}}%
\thanks{This work has been supported by the Spanish project Ministerio de Econom\'ia y Competitividad DPI2015-69376-R  and partially supported by the European Commission under the Grant Agreement number 774571 (Project PANTHEON -- ``Precision farming of hazelnut orchards'').}
}

\begin{document}

\maketitle
\thispagestyle{empty}
\pagestyle{empty}

\begin{abstract}
Measures of node centrality that describe the importance of a node  within a network are crucial for understanding the behavior of social networks and graphs. In this paper, we address the problems of distributed estimation and control of node centrality in undirected graphs with asymmetric weight values.
In particular, we focus our attention on $\alpha$-centrality, which can be seen as a generalization of eigenvector centrality.
In this setting, we first consider a distributed protocol where agents compute their \mbox{$\alpha$-centrality}, focusing on the convergence properties of the method; 
then, we combine the estimation method with a consensus algorithm to achieve a consensus value weighted by the influence of each node in the network.
Finally, we formulate an \mbox{$\alpha$-centrality} control problem which is naturally decoupled and, thus, suitable for a distributed setting and we apply this formulation to  protect the most valuable nodes in a network against a targeted attack, by making every node in the network equally important in terms of \mbox{$\alpha$-centrality.}
Simulations results are provided to corroborate the theoretical findings. 
\end{abstract}

\section{Introduction}\label{sec:intro}
In graph theory and network analysis, identifying the most {\em central} nodes, i.e., the most important nodes within a graph,  has been a very important research topic for a long time, see~\cite{Freeman:1978,Newman:2010}.
Applications involving centrality concepts include, among others,
identifying the most influential person(s) in a social network, finding key infrastructure nodes in the Internet or urban networks, and pinpointing super-spreaders of disease.  Depending on the specific domain of interest, a variety of metrics have been proposed to measure the centrality of nodes in a network,  ranging from node degree~\cite{wasserman1994social}, eccentricity~\cite{oliva2016distributed}, closeness~\cite{Bavelas:1950} and betweenness~\cite{Freeman:1977} to eigenvector centrality~\cite{Bonacich:1972} and \mbox{$\alpha$-centrality~\cite{Bonacich:2001}.}

Recently, a few works which attempt to compute node centrality in a distributed fashion have been presented to the research community.
In~\cite{Lehmann:2003}, a framework for the calculation of the betweenness-centrality is proposed. 
In~\cite{Wehmuth:2012}, a distributed method is given to assess network closeness-centrality based only on localized information restricted to a given neighborhood around each node.  
In~\cite{Tang:2013,Tang:2014}, different distributed algorithms to compute betweenness and closeness centrality in a tree graph are proposed. In~\cite{Tang:2015}, the authors extend their previous results on closeness centrality to the case of general graphs, by formulating a set of linear inequality and equality constraints, which are distributed in nature. 
The distributed estimation of betweeness centrality is exploited in~\cite{maccari2018distributed} to design an efficient routing protocol.
The authors of~\cite{Rossi:18} present a distributed method for the estimation of the Harmonic influence centrality~\cite{Rossi:16}, defined in the context of opinion dynamics.

One of the most famous measurements of centrality in networks is the PageRank, which is a modified version of eigenvector-centrality, of special interest in web ranking~\cite{Ishii:14}.
In~\cite{Charalambous:2016:CDC}, the authors propose a distributed algorithm for computing the eigenvector centrality, which accounts for both the lack of synchronicity and heterogeneity of agents in terms of clock rates.
In~\cite{Tempo:2017}, the authors propose deterministic finite-time algorithms  for measuring degree, closeness, and betweenness centrality, along with a randomized algorithm for computing the PageRank. 
Recently, distributed Page-Rank estimation was computed by means of the Power method in~\cite{Suzuki:18}.
It is worth mentioning that all these methods focus on the estimation of the centrality, but none of them considers the problem of controlling it, i.e., introducing control mechanisms to drive the centrality value to a specific state.

In this work, we focus our attention on the  \hbox{$\alpha${\em-centrality}}~\cite{Bonacich:2001}, which can be seen as a  generalization of eigenvector centrality that is particularly suitable for networks with asymmetric interactions. Briefly, $\alpha$-centrality measures the total number of paths from a node, exponentially attenuated by their length, where the parameter $\alpha$  sets the length scale of interactions.  
Compared to other centrality metrics, an interesting property of $\alpha$-centrality is that it can be a tool to discriminate between locally and globally connected nodes; by locally connected nodes we mean nodes that are part of a community, in that their neighbors exhibit a large degree of mutual interconnection, while by globally connected nodes we mean nodes that interconnect poorly connected groups of nodes. Notably, studies on human beings~\cite{Granovetter:1973} and animals~\cite{LusseauS477} have provided evidence that these latter nodes, often recognized as ``bridges'' or ``brokers'', play a crucial role in the information flow cohesiveness of the entire group.

Our contribution is then threefold: i) we describe a distributed protocol where agents, by means of local interactions, locally compute their \hbox{$\alpha$-centrality} \black{indices}
and accurately characterize its dynamic behavior; 
ii) we combine the estimation method with a consensus algorithm to achieve a consensus value weighted by the influence of each node in the network.
iii) we formulate an \mbox{$\alpha$-centrality} control problem which is naturally decoupled and, thus, suitable for a distributed setting, exploiting such formulation to protect the most valuable nodes in a network against a targeted attack.
A preliminary version of this paper appeared in~\cite{EM-GO-AG:18}.
Compared to it, we have included additional details on the estimation procedure, as well as lifted an assumption on the weighted consensus application. The control solution has not been published before. 

The rest of the paper is organized as follows. In Section~\ref{sec:preliminaries} some background notions are provided. 
In Section~\ref{sec:estimation}, the proposed distributed algorithm
to compute the $\alpha$-centrality index through local interactions is described. 
In Section~\ref{sec:influence} we describe a novel weighted consensus algorithm based on $\alpha$-centrality.
In Section~\ref{sec:control}
we discuss an optimization problem that can be solved locally to allow the agents reach a desired $\alpha$-centrality.
In Section~\ref{sec:simulations} we present simulation results and in Section~\ref{sec:conclusions} the main conclusions of this work.

\section{Preliminaries}
\label{sec:preliminaries}
Let us consider a network of $N$ nodes labeled by $i \in \mathcal{V}$.   The nodes exchange information with each other following a fixed {\em undirected} communication graph $\mathcal{G}=(\mathcal{V},\mathcal{E})$, where $\mathcal{E} \subset \mathcal{V} \times \mathcal{V}$ represents the edge set.  In this way, nodes $i$ and $j$ can communicate if and only if $(i,j) \in \mathcal{E}$. We assume that the communication graph is connected. The set $\mathcal{N}_i$ of neighbors of node $i\in\mathcal{V}$ is the subset of nodes that can directly communicate with it, i.e., $\mathcal{N}_i = \{j \in \mathcal{V} \;|\; (i,j) \in \mathcal{E}\}$.

Given a graph \mbox{$\mathcal{G}=(\mathcal{V},\mathcal{E})$}, let us define the set of matrices compatible
with $\mathcal{G}$ as
\begin{equation*}
   \mathbb{A}_{\mathcal{G}}=\left \{\mathbf{W}\in \mathbb{R}^{\black{N\times N}}\big|\,  w_{ij}=0, \quad \forall \,  (i,j)\not\in \mathcal{E}\right \}.
\end{equation*}
In particular, we define $\mathbf{W}\in\mathbb{A}_{\mathcal{G}}$ as the influence matrix associated to the network, a weighted adjacency matrix where $w_{ij}\ge 0$ represents the influence that agent $j$'s information has for agent $i$.
Note that the influence matrix can be {\em asymmetric}, to model the fact that two neighboring agents can place different importance on the information provided by each other. 
The influence matrix can also contain values equal to zero between neighbors, meaning that agent~$j$ is completely disregarded by agent~$i$, even when they talk to each other.
Finally, non-neighboring agents have mutual zero influence by construction, since they do not communicate. 

The notion of node centrality in graph theory is built around the idea of measuring how important  a particular vertex is over a certain graph structure. This is typically expressed in terms of a function, $\black{\beta}(\mathbf{W}(\mathcal{G})): \mathbb{R}^{N\times N} \rightarrow \black{\mathbb{R}_{\ge 0}},$ 
that computes a vector where the \mbox{$i$-th} component expresses the importance of node $i$ in the whole network structure, characterized by the matrix $\mathbf{W}$. Note that this function can also be computed for any other weight matrix associated to the graph.

Among the multiple possibilities for computing node centrality, we will focus on the $\alpha$-centrality~\cite{Bonacich:2001}, which is expressed
\begin{equation}
\label{Eq:a-centrality}
    \bm \rho_\alpha = \left(\mathbf{I}_{N}-\alpha\mathbf{W}^T \right)^{-1}\mathbf{z},
\end{equation}
where $\alpha$ is a parameter that measures how the importance fades away with distance in communication hops and $\mathbf{z}$ is an arbitrary vector, with non-negative values and at least one positive value, that can be used to provide the nodes with some initial importance. 

It is noteworthy that setting $\mathbf{z} = \mathbf{1}_{N}$ we can easily extend our algorithms to deal with an important metric in networks, namely, the {\em Katz-centrality}~\cite{Newman:2010},
which is defined as
\begin{equation}
\label{Eq:katz-centrality}
    \bm \rho_K = \left(\left(\mathbf{I}_{N}-\alpha\mathbf{W}^T \right)^{-1}-\mathbf{I}_{N}\right)\mathbf{1}_{N}.
\end{equation}

\begin{assumption}
\label{Ass:alpha}
The parameter $\alpha>0$ is such that \mbox{$\rho(\mathbf{W})\alpha < 1,$} where $\rho(\mathbf{W})$ denotes the spectral radius of~$\mathbf{W},$~\cite{Bonacich:2001}.
\end{assumption}

\black{This assumption is required to ensure that
$\bm\rho_\alpha$ in Eq.~\eqref{Eq:a-centrality} is well defined and belongs to $\mathbb{R}_{\ge 0}^N.$ 
More specifically, 
 \mbox{$\rho(\mathbf{W})\alpha < 1$} guarantees the existence of the inverse in Eq.~\eqref{Eq:a-centrality}, and this 
allows us to rewrite the expression as a sum of positive terms, since $\alpha,\mathbf{W}$ and $\mathbf{z}$ are non-negative, leading to the second claim.}

\section{Distributed Estimation of Node $\alpha$-Centrality}
\label{sec:estimation}

In this section we present a distributed linear iteration protocol that allows each agent to compute its own value of $\alpha$-centrality.
Let $c_i(t)$ be the estimation that agent $i$ has of the $i$-th component of $\bm \rho_\alpha$ at iteration $t$, with arbitrary initial conditions and which is updated by
\begin{equation}
\label{Eq:Dist_acentralityi}
    c_i(t+1) = \alpha \sum_{j\in {\mathcal{N}}_i} w_{ji}c_j(t) + z_i.
\end{equation}
The protocol can also be expressed in vectorial form as
\begin{equation}
\label{Eq:Dist_acentralityVec}
    \mathbf{c}(t+1) = \alpha \mathbf{W}^T \mathbf{c}(t) + \mathbf{z}.
\end{equation}

Note that, a slightly modified version of this algorithm was originally proposed in~\cite{Newman:2010} and in~\cite{Suzuki:18} for a centralized setup and for distributed Page-Rank estimation, respectively. In this regard, we also provide an accurate characterization of the convergence rate of this algorithm, which was not discussed in any of the aforementioned papers.

It is noteworthy that, compared to other typical linear protocols, Eq.~\eqref{Eq:Dist_acentralityi} uses $w_{ji}$ instead of $w_{ij}$.
Although mathematically there are no differences in using $\mathbf{W}$ instead of $\mathbf{W}^T,$ there is an interesting motivation for considering this.
If the algorithm is used with $w_{ij}$ instead of $w_{ji}$, the final outcome would be a measurement of how influenced a node is by the rest of the network. In real applications, e.g., Twitter, there is more interest in knowing how a particular user can affect the network than in knowing the opposite, which is the motivation for this difference with respect to literature.

In order to characterize the convergence properties of the algorithm, let $e_c(t) = \|\mathbf{c}(t) -\bm \rho_\alpha\|$
be the estimation error for the whole centrality vector at iteration $t$.

\begin{theorem}
\label{Theorem1}
For any matrix $\mathbf{W}$  and any value of $\alpha$ that satisfies Assumption~\ref{Ass:alpha}, the execution of~\eqref{Eq:Dist_acentralityVec} leads to the distributed computation of node $\alpha$-centrality, i.e.,
\begin{equation}
\label{Eq:DistCentrality_conv}
    \lim_{t\to\infty} \mathbf{c}(t) = \bm \rho_\alpha.
\end{equation}
Besides, 
the estimation error, $e_c(t),$ is upper-bounded by 
\begin{equation}
\label{Eq:errorBound}
    e_c(t)\le \gamma\dfrac{(2-\kappa)\kappa^{t}}{1-\kappa} 
    \max(\left\| \mathbf{c}(0)\right\|_\mathbf{W},\left\| \mathbf{z}\right\|_\mathbf{W}),
\end{equation}
\black{with $0\le\kappa<1$ and $\gamma>0$ constants and $\| \cdot \|_\mathbf{W}$ a norm, that depends on $\mathbf{W},$ such that $\alpha\|\mathbf{W}\|_\mathbf{W}<1$.}
\end{theorem}
\begin{proof}
First of all, from Assumption~\ref{Ass:alpha}, we can express the centrality vector as a series
\begin{equation}
\label{Eq:rhoSeries}
    \bm\rho_\alpha  =   \sum_{k=0}^{\infty}(\alpha\mathbf{W}^T)^{k}\mathbf{z}.
\end{equation}
Note that~\eqref{Eq:Dist_acentralityVec} corresponds to 
\begin{equation}
\label{Eq:centrality_induction}
    \mathbf{c}(t) = \left(\alpha\mathbf{W}^T\right)^{t}\mathbf{c}(0) + \sum_{k=0}^{t-1} \left(\alpha\mathbf{W}^T\right)^k\mathbf{z}. 
\end{equation}
Therefore, when $t$ goes to infinity, by Assumption~\ref{Ass:alpha} the first term of the equation goes to zero and the second approaches to~\eqref{Eq:rhoSeries}, 
thus proving the convergence of our algorithm.

\black{Regarding the error, by Assumption~\ref{Ass:alpha}, we know that $\alpha\rho(\mathbf{W})<1$. By resorting to  Theorem 6.9.2 of~\cite{Stoer92}, we can assure that there exists a norm, $\| \cdot \|_\mathbf{W}$, dependent on $\mathbf{W}$, such that $\alpha\|\mathbf{W}\|_\mathbf{W}<1$. Besides, for any two norms 
$\|\cdot \|_a$ and
$\|\cdot \|_b$ we can always find a positive constant $\gamma$  such that 
for any vector $\mathbf{v} \in  \mathbb{
R}^N$ it holds $\|\mathbf{v}\|_a < \gamma \|\mathbf{v}\|_b$.
Then, using again Eqs.~\eqref{Eq:rhoSeries} and~\eqref{Eq:centrality_induction}, }
\begin{equation}
    \begin{aligned}
        &e_c(t) = \|\mathbf{c}(t) -\bm \rho_\alpha\| \le \gamma\|\mathbf{c}(t) -\bm \rho_\alpha\|_\mathbf{W} \\
         &= \gamma\left\| \left(\alpha\mathbf{W}^T\right)^{t}\mathbf{c}(0) - \sum_{k=t}^{\infty} \left(\alpha\mathbf{W}^T\right)^k\mathbf{z}
         \right\|_\mathbf{W} \\
         &\le\gamma \black{\left( \left\| \left(\alpha\mathbf{W}^T\right)^{t}\right\|_\mathbf{W} + \sum_{k=t}^{\infty} \left\|\left(\alpha\mathbf{W}^T\right)^k\right\|_\mathbf{W}\right)}
        \max(\mathbf{c},\mathbf{z})\\
        &\le 
        \gamma\left(\kappa^t+\sum_{k=t}^{\infty}\kappa^k\right) \max(\mathbf{c},\mathbf{z}),
     \end{aligned}
\end{equation}
where, for clarity, $\max(\mathbf{c},\mathbf{z})$ is an abbreviation for $\max(\left\| \mathbf{c}(0)\right\|_\mathbf{W},\left\| \mathbf{z}\right\|_\mathbf{W})$ and $\kappa = \alpha\|\mathbf{W}\|_\mathbf{W} < 1.$
Finally, applying the properties of geometric series the bound in Eq.~\eqref{Eq:errorBound} is obtained.
\end{proof}

Let us now discuss the influence of the parameter $\alpha$ on the estimation procedure. 
According to Assumption~\ref{Ass:alpha} it follows that the parameter $\alpha$ needs to be ``small enough'' to ensure the convergence of the estimation algorithm. In order to obtain a distributed procedure for computing such parameter $\alpha$, it should be noticed that it holds
\begin{equation*}
    \|\mathbf{W}\| \le \sqrt{\|\mathbf{W}\|_1 \|\mathbf{W}\|_\infty};
\end{equation*}
therefore, by choosing 
\begin{equation}
    \alpha<1/(\sqrt{\|\mathbf{W}\|_1 \|\mathbf{W}\|_\infty})
\end{equation}
it follows that 
\begin{equation*}
    \alpha\rho(\mathbf{W}) \black{ \le } \alpha \|\mathbf{W} \| \black{ \le }
    \dfrac{\|\mathbf{W} \|}{\sqrt{\|\mathbf{W}\|_1 \|\mathbf{W}\|_\infty}} < 1.
\end{equation*}
thus Assumption~\ref{Ass:alpha} holds by construction.
At this point, we observe that 
 the communication graph is undirected and the agents have knowledge of the entries $w_{ij}$ and $w_{ji}$ corresponding to their neighbors. Therefore, both the one norm and the infinity norm of $\mathbf{W}$ can be easily computed in finite time via a max-consensus (leader election) protocol~\cite{lynch1996distributed}. As a consequence, if needed, the network can agree in a distributed fashion upon a suitable value of the parameter $\alpha$.

\section{Influenced-based weighted consensus}
\label{sec:influence}
In this section we present a consensus algorithm that is able to reach a weighted consensus on some initial conditions accounting for the influence that each node has in the network.
The algorithm can be used in cooperative estimation problems, where the degree of confidence of the different nodes is encoded in the influence matrix, thus weighting more the opinion of important nodes.
A peculiarity of our algorithm is that it does not require to know beforehand the actual value of the network $\alpha$-centrality. Instead, the algorithm adjusts the consensus value according to the current value of the $\alpha$-centrality vector, estimated in parallel in a distributed fashion.

Let $x_i(0)$ be the initial condition of agent $i$ to be incorporated in the consensus iteration and $\mathbf{x}(0)$ the concatenation of the initial conditions of all the agents in vector form.
Compared to the typical consensus problem of reaching the average of the initial conditions, our aim is to compute in a distributed fashion the following quantity,
\begin{equation}
\label{Eq:weightedConsensus}
    x^* = \dfrac{\bm\rho_\alpha^T \mathbf{x}(0)}{\bm\rho_\alpha^T \mathbf{1}},
\end{equation}
which is a weighted average based on the global influence that each agent has in the network, according to the \mbox{$\alpha$-centrality} vector. In this way, more influential agents will have larger weight in the final consensus value than those with less influence power.
While there exist algorithms that, knowing $\bm\rho_\alpha$, are able to compute Eq.~\eqref{Eq:weightedConsensus}, (e.g., see~\cite{Sundaram:2008}), 
our objective is to compute this value without prior knowledge of the centrality vector by the network.

\black{In order to do this, let us start by defining the Perron matrix, $\mathbf{Q}=[q_{ij}],$ associated to the Laplacian matrix $L$ of $\mathcal{G}$, i.e., \mbox{$\mathbf{Q} = \mathbf{I} - \varepsilon L(\mathcal{G}), \varepsilon>0$}.
Assume $\varepsilon$ is small enough to ensure that $\mathbf{Q}$ is symmetric and doubly stochastic matrix\footnote{$\varepsilon<2/N$ ensures this assumption. 
} with largest eigenvalue equal to one and the second largest (in absolute value), denoted by $\lambda_{\mathbf{Q}}<1$.}
For the sake of completeness, recall that, under the above assumptions on $\varepsilon$, the classical linear iteration
\begin{equation}
\label{eq:classicCons}
    \mathbf{x}(t+1) = \mathbf{Q} \mathbf{x}(t),
\end{equation}
asymptotically converges to the average of the initial conditions, $\mathbf{x}(0)$,~\cite{FB-JC-SM:09}.

The high level idea of our algorithm consists of applying an exogenous input, $\gamma_i,$ to the classical consensus algorithm, such that the new final value corresponds to Eq.~\eqref{Eq:weightedConsensus},
\begin{equation}
\label{Eq:ConsensusInput}
    x^*=\dfrac{1}{N}\sum_{i\in\mathcal{V}} \left(x_i(0) + \gamma_i \right). 
\end{equation}
Combining Eq.~\eqref{Eq:weightedConsensus} and Eq.~\eqref{Eq:ConsensusInput}, the value of this input needs to be
\begin{equation}
\label{Eq:correctionTerm}
    \gamma_i = \left( \dfrac{N\rho_i}{\sum_{j=1}^N \rho_j} - 1\right)x_i(0) = \left( \dfrac{\rho_i}{\bar{\rho}_\alpha} - 1\right)x_i(0)
\end{equation}
where $\rho_i$ represents the $i$-th component of $\bm\rho_\alpha$ and $\bar{\rho}_\alpha$ represents the average of the
influence weights $\bm\rho_\alpha$, that is 
\begin{equation*}
\bar{\rho}_\alpha =\dfrac{1}{N} \sum_{j=1}^N \rho_{j}.
\end{equation*}

However, note that, as mentioned before, agents do not have the knowledge of $\bm\rho_\alpha$ nor of its average, $\bar{\rho}_\alpha$.
Thus, our proposed algorithm consists of the following cascading update rules,
\begin{subequations}
\label{Eq:WeightedCons_Ind}
     \begin{align}
        c_i(t+1) &= \alpha \sum_{j\in {\mathcal{N}}_i} w_{ji}c_j(t) + z_i, \label{eq:cascade:a}\\
        \Delta c_i (t+1) &= c_i(t+1) - c_i(t),\\
        \label{Eq:WconsInd3}
        \bar{c}_i(t+1) &= \sum_{j\in \black{\mathcal{N}_i}}q_{ij}\bar{c}_j(t) + \Delta c_i(t+1),\\
        \label{Eq:yi}
        {y}_i(t+1) &= \left(\dfrac{c_i(t+1)}{\bar{c}_i(t+1)}-1\right) {x}_i(0),\\
        \label{Eq:deltayi}
        \Delta y_i (t+1) &= y_i(t+1) - y_i(t),\\
        x_i(t+1) &= \sum_{j\in \black{\mathcal{N}_i}}q_{ij}x_j(t)+\Delta y_i (t+1),\label{eq:cascade:f}
     \end{align}
\end{subequations}
with \black{initial conditions such that $c_i(0)=\bar{c}_i(0)=z_i$}, $y_i(0)=0$ 
and $x_i(0)$ the initial consensus value of agent $i$ as defined at the beginning of the section.
The same algorithm can be expressed in vectorial form as,
\begin{subequations}
\label{Eq:WeightedCons_Vec}
    \begin{align}
    \label{Eq:CentEst_vec1}
    \mathbf{c}(t+1) &= \alpha \mathbf{W}^T \mathbf{c}(t) + \mathbf{z},\\
    \label{Eq:CentEst_vec2}
    \Delta\mathbf{c}(t+1) &= \mathbf{c}(t+1)-\mathbf{c}(t),\\
    \label{Eq:WeightCons_AvgCent}
        \bar{\mathbf{c}}(t+1) &= \mathbf{Q}\bar{\mathbf{c}}(t) + \Delta\mathbf{c}(t+1),\\
        \label{Eq:WeightCons_ytvect}
        \mathbf{y}(t+1) &= \diag \left(\dfrac{c_i(t+1)}{\bar{c}_i(t+1)}-1\right)\mathbf{x}(0),\\
        \label{Eq:WeightCons_DeltaY}
        \Delta \mathbf{y} (t+1) &= \mathbf{y}(t+1) - \mathbf{y}(t),\\
        \label{Eq:WeightCons_Infl}
        \mathbf{x}(t+1) &= \mathbf{Q}\mathbf{x}(t)+\Delta \mathbf{y} (t+1).
    \end{align}
\end{subequations}
The intuition behind each rule is the following:
Eq.~\eqref{Eq:CentEst_vec1}, equivalent to Eq.~\eqref{Eq:Dist_acentralityVec},
is used for the distributed computation of $\bm\rho_\alpha$ and included here for completeness of the rule.
Eq.~\eqref{Eq:CentEst_vec2} tracks the changes in the estimation of $\bm\rho_\alpha$.
Eq.~\eqref{Eq:WeightCons_AvgCent} intends to compute the average value of $\bm\rho_\alpha$, 
estimated in $\mathbf{c}(t)$. 
The addition of $\Delta\mathbf{c}(t+1)$ is necessary to account for the estimation error made in $\mathbf{c}(t+1).$
The vector $\mathbf{y}(t)$ aims at computing Eq.~\eqref{Eq:correctionTerm}.
However, since the convergence to the correct value is asymptotic with $\mathbf{Q}$, instead of applying this input at once, we apply it incrementally at each communication round in Eq.~\eqref{Eq:WeightCons_Infl}, similarly to what was done in~\cite{Priolo:2014} to compute the average in an unbalanced digraph.

Before analyzing the convergence properties of the cascade system, we introduce the following Lemma, to handle the possible case of iterations where, $\bar{c}_i(t) = 0.$
\begin{lemma}
\label{lema3}
Suppose that all the entries of the vector $\mathbf{z}$ are non-negative and at least one is strictly positive. Then there exists some $t^*$ such that for all $t>t^*$ all the components in $\bar{\mathbf{c}}(t)$ are strictly positive.
\end{lemma}
\begin{proof}
First of all, \black{by imposing initial condition $c_i(0)=z_i$ implies that that $c_i(t)$ is an increasing function with $t$, and consequently $\Delta c_i(t)$ is not negative. This can be demonstrated transforming Eq.~\eqref{Eq:CentEst_vec1} into the equivalent form
\begin{equation*}
\mathbf{c}(t+1) = \mathbf{c}(t) + \alpha \mathbf{W}^T \Delta \mathbf{c}(t),
\end{equation*}
As a consequence, this term in Eq.~\eqref{Eq:WeightCons_AvgCent} is only additive.}
This means that if the claim is true without considering this term, then it will also hold including it. Thus, let us assume that $\Delta c_i(t)=0$ for all $i$ and all $t$.
This implies that Eq.~\eqref{Eq:WeightCons_AvgCent} becomes a classic averaging rule as in Eq.~\eqref{eq:classicCons}.
Denote \mbox{$\bar{z}=\sum_i z_i / N$}.
Since all the elements in $\mathbf{z}$ are non-negative and at least one is positive we can assert that $\bar{z}>0.$
Now, we know that for all $\Delta c_i(t)=0,$ Eq.~\eqref{Eq:WeightCons_AvgCent} will converge to the average of the initial condition $\mathbf{c}(0)$ which in this case is equal to $\mathbf{z}$, so $\bar{\mathbf{c}}(t)$ will converge asymptotically to $\bar{z}\mathbf{1}.$
This implies that for any arbitrarily small $\epsilon>0$ we can find a $t^*$ such that for all $t>t^*,$ for all $i$ it holds that $|\bar{c}_i(t)-\bar{z}|<\epsilon$.

Consequently, there will be a time, $t^*$ such that for all $t>t^*,$ all the components $\bar{\mathbf{c}}(t)$ will be strictly positive, completing the proof.
\end{proof}

It should be noticed that Lemma~\ref{lema3} is necessary to provide an algorithmic implementation of the proposed protocol. 
As a matter of fact, by looking at the cascade of update rules given in Eqs.~\eqref{eq:cascade:a}-\eqref{eq:cascade:f}, it can be noticed that if $\bar{c}_i(t)=0$, then Eq.~\eqref{Eq:yi} is not defined. In order to overcome this issue,  Eq.~\eqref{Eq:yi} can be replaced as
\begin{equation}
\label{Eq:Ywelldef}
    y_i(t+1) = 
    \begin{cases}
    \hspace{2mm} y_i(t) & \hbox{if }  \bar{c}_i(t+1) = 0\\[4pt]
    \left(\dfrac{c_i(t+1)}{\bar{c}_i(t+1)}-1\right) {x}_i(0) & \hbox{otherwise}  
    \end{cases}
\end{equation}
As it will be shown later in Theorem~\ref{Theorem2}, this change does not affect the overall convergence of the algorithm. Intuitively, this can be explained by the fact that our goal is to apply the total input by means of a sequence of increments, where at each iteration we compensate for the error in the estimation of the centrality. Therefore, by adding and subtracting the same quantity we do not modify the total input while avoiding the division by zero in Eq.~\eqref{Eq:WeightCons_ytvect}. For the sake of clarity, the equivalent vectorial version of Eq.~\eqref{Eq:WeightCons_ytvect} based on Eq.~\eqref{Eq:Ywelldef} is here omitted.

Let us now review an auxiliary result used to prove the main Theorem of this section, convergence of our algorithm.
\begin{lemma}[Lemma 3.2 in~\cite{Montijano:2014}]
\label{lema2}
Let $0 \le\lambda <1$ and $\{\beta(t)\}$ a bounded sequence
such that $\lim_{t\to \infty} \beta(t) =0$. Then
\begin{equation*}
\lim_{t\to\infty} \sum_{j=0}^t \lambda^{t-j}\; \beta(j) =0.
\end{equation*}
\end{lemma}

\begin{theorem}
\label{Theorem2}
Assume the conditions in Theorem~\ref{Theorem1} hold;
then, the dynamical system in Eq.~\eqref{Eq:WeightedCons_Vec} converges to
\begin{subequations}
\label{Eq:ConvergenceCon}
    \begin{align}
    \label{Eq:ConvergenceCon0}
    &\lim_{t\to\infty}\mathbf{c}(t) = \bm\rho_\alpha,\\
    \label{Eq:ConvergenceCon1}
    &\lim_{t\to\infty}\Delta\mathbf{c}(t) = \mathbf{0},\\
    \label{Eq:ConvergenceCona}
    &\lim_{t\to\infty}\bar{\mathbf{c}}(t) = \bar{\rho}_\alpha\mathbf{1},\\
    \label{Eq:ConvergenceConb}
    &\lim_{t\to\infty}\mathbf{y}(t) = 
    \diag \left(\gamma_i\right)\mathbf{1},\\
    \label{Eq:ConvergenceCon2}
    &\lim_{t\to\infty}\Delta\mathbf{y}(t) = \mathbf{0},\\
    \label{Eq:ConvergenceConc}
    &\lim_{t\to\infty}\mathbf{x}(t) = x^*\mathbf{1}.
    \end{align}
\end{subequations}
\end{theorem}
\begin{proof}
The limit in Eq.~\eqref{Eq:ConvergenceCon0} was already demonstrated in Theorem~\ref{Theorem1} and the limit in Eq.~\eqref{Eq:ConvergenceCon1} comes naturally from it.
Let us define now the average of the centrality estimation increments, i.e., \mbox{$\Delta\bar{c}(t) = \sum_{i\in\mathcal{V}} \Delta c_i(t)/N$,} \black{and the auxiliary symbol, $\bar{c}_0=1/N\sum_i c_i(0)$.}
The average of the centrality vector can be expressed as an infinite sum, and, using Eq.~\eqref{Eq:ConvergenceCon0}, it holds
\black{
\begin{equation}
\label{Eq:barRhoalphasum}
\begin{aligned}
\sum_{t=1}^{\infty} \Delta\bar{c}(t)&=\sum_{t=1}^{\infty}\dfrac{1}{N}\sum_{i=1}^{N} \Delta c_i(t)
=\dfrac{1}{N}\sum_{i=1}^{N}\sum_{t=1}^{\infty}\Delta c_i(t)\\
&=\dfrac{1}{N}\sum_{i=1}^{N}\left(\lim_{t\to\infty} c_i(t)-c_i(0)\right)\\
&=\dfrac{1}{N}\sum_{i=1}^{N} \left(\rho_{i}-c_i(0)\right)=\bar{\rho}_\alpha - \bar{c}_0,
\end{aligned}
\end{equation}
}
In addition, 
$\bar{\mathbf{c}}(t)$ can be expressed as a sum as follows
\black{
\begin{equation}
\label{Eq:barCtsum}
    \bar{\mathbf{c}}(t) = \mathbf{Q}^{t}\bar{\mathbf{c}}(0) + \sum_{j=1}^{t} \mathbf{Q}^{t-j}\Delta\mathbf{c}(j),
\end{equation}
}
Let us define now the difference, \mbox{$e_{\bar{c}}(t)=\|\bar{\mathbf{c}}(t)-\bar{\rho}_\alpha\mathbf{1}\|,$} and compute its limit when the time goes to infinity,
\black{
\begin{equation}
\label{Eq:errDevel1}
\kern -10pt\begin{aligned}
    &\lim_{t\to\infty} e_{\bar{c}}(t) = \lim_{t\to\infty}\left\|\bar{\mathbf{c}}(t)-\bar{\rho}_\alpha\mathbf{1}\right\| \\
   & = \lim_{t\to\infty}\left\|\mathbf{Q}^t\bar{\mathbf{c}}(0)-\bar{c}_0\mathbf{1}+\sum_{j=1}^{t} \left(\mathbf{Q}^{t-j}\Delta\mathbf{c}(j)-\Delta\bar{c}(j)\mathbf{1}\right)\right\| \\
    & \le \lim_{t\to\infty}\sum_{j=1}^{t} \left\|\mathbf{Q}^{t-j}\Delta\mathbf{c}(j)-\Delta\bar{c}(j)\mathbf{1}\right\|, \\
\end{aligned}
\end{equation}
where the second line comes from replacing Eq.~\eqref{Eq:barRhoalphasum} and Eq.~\eqref{Eq:barCtsum} and the third one is by direct application of norm inequalities, plus accounting that $\lim_{t\to\infty}\mathbf{Q}^t\bar{\mathbf{c}}(0)=\bar{c}_0\mathbf{1}$.
}

Before proceeding, we recall that, for all $t\ge 0,$ it holds
\begin{equation}
\mathbf{Q}^t=\dfrac{1}{N}{\bm 1}\cdot{\bm 1}^T+\sum_{i=2}^N \lambda_{\mathbf{Q},i}^{t}{\bm v}_{\mathbf{Q},i}{\bm v}_{\mathbf{Q},i}^T
\end{equation}
where $\lambda_{\mathbf{Q},i}$ is the eigenvalue of $\mathbf{Q}$ with $i$-th largest magnitude and $v_{\mathbf{Q},i}$ is the associated eigenvector. 
Since $\mathbf{Q}$ is symmetric, we also know that
${\bm v}_{\mathbf{Q},i}^T{\bm 1}=0$, $i\geq 2$.
Thus, following the development of Eq.~\eqref{Eq:errDevel1},
\begin{equation}
\label{Eq:errDevel2}
\begin{aligned}
    \lim_{t\to\infty} e_{\bar{c}}(t) &\le 
\lim_{t\to\infty}\sum_{j=1}^{t} \left\|\sum_{i=2}^N \lambda_{\mathbf{Q},i}^{t-j}{\bm v}_{\mathbf{Q},i}{\bm v}_{\mathbf{Q},i}^T\Delta\mathbf{c}(j)\right.
\\
&\hspace{20mm}\left.+\dfrac{{\bm 1}\cdot{\bm 1}^T }{N}\left(\Delta\mathbf{c}(j)-\Delta\bar{c}(j)\mathbf{1}\right)\right\| \\
    & \le \lim_{t\to\infty}\gamma_\mathbf{Q}\sum_{j=1}^{t} \lambda_{\mathbf{Q},2}^{t-j}\left\|\Delta\mathbf{c}(j)\right\|,
\end{aligned}
\end{equation}
with $\gamma_\mathbf{Q}=N\max_{i\neq 1}\|{\bm v}_{\mathbf{Q},i}{\bm v}_{\mathbf{Q},i}^T\|$ 
a constant.
Finally, using Theorem~\ref{Theorem1} we know that $\|\Delta\mathbf{c}(j)\|$ is bounded and converges to zero as $j$ goes to infinity. Additionally, we know that $0\le\lambda_{\mathbf{Q},2}<1.$ Thus, using Lemma~\ref{lema2} we can assert that $e_{\bar{c}}$ converges to zero, showing that Eq.~\eqref{Eq:ConvergenceCona} is true.


Once we have established convergence of $\bar{\mathbf{c}}(t)$, combining this limit with Theorem~\ref{Theorem1} together with Eq.~\eqref{Eq:correctionTerm} and Lemma~\ref{lema3}, the limit presented in Eq.~\eqref{Eq:ConvergenceConb} follows up straightforwardly and, consequently, so does the one in Eq.~\eqref{Eq:ConvergenceCon2}.


In order to prove Eq.~\eqref{Eq:ConvergenceConc}, let us notice that by combining Eq.~\eqref{Eq:ConsensusInput} together with Eq.~\eqref{Eq:correctionTerm} and recalling that $y_i(0)=0,$ $\forall \, i \,  \in \, \mathcal{V}$, we obtain 
\begin{equation} \label{eq:derivation:1}
\begin{aligned}
    x^* 
    &=\dfrac{1}{N}\sum_{i\in\mathcal{V}} \left(x_i(0) + \displaystyle\lim_{t\to\infty} y_i(t) \right) \\
    &=\dfrac{1}{N}\sum_{i\in\mathcal{V}} \left(x_i(0) +  \sum_{t=1}^{\infty} (y_i(t)-y_i(t-1)) \right) \\
    &=\dfrac{1}{N}\sum_{i\in\mathcal{V}} x_i(0) +   \sum_{t=1}^{\infty}
    \dfrac{1}{N}\sum_{i\in\mathcal{V}} \Delta y_i(t))  \\
     &=\dfrac{1}{N}\sum_{i\in\mathcal{V}} x_i(0) + \sum_{t=1}^{\infty} \Delta \bar{y}(t) \\
\end{aligned},
\end{equation}
where, similarly to the case in Eq.~\eqref{Eq:ConvergenceCona}, the following definition has been used
\begin{equation} \label{eq:bbb}
\Delta\bar{y}(t) = \dfrac{1}{N}\sum_{i\in\mathcal{V}}\Delta y_i(t).
\end{equation}
%

At this point, we observe that $\mathbf{x}(t)$ can be written as
\begin{equation}
    \mathbf{x}(t) = \mathbf{Q}^{t}\mathbf{x}(0) + \sum_{j=1}^{t} \mathbf{Q}^{t-j}\Delta\mathbf{y}(j);
\end{equation}
hence, it follows that
\begin{equation} \label{eq:last_derivation_convergence}
\begin{aligned}
\lim_{t\rightarrow \infty }\mathbf{x}(t) &= \lim_{t\rightarrow \infty }\left(\mathbf{Q}^{t}\mathbf{x}(0) + \sum_{j=1}^{t} \mathbf{Q}^{t-j}\Delta\mathbf{y}(j)\right)\\
&=\dfrac{1}{N}{\bm 1}\cdot{\bm 1}^T {\bm x}(0)+
\lim_{t\rightarrow \infty } \sum_{j=1}^{t} 
\dfrac{1}{N}{\bm 1}\cdot{\bm 1}^T\Delta\mathbf{y}(j)\\
&+\lim_{t\rightarrow \infty } \sum_{i=2}^N \lambda_{\mathbf{Q},i}^{t-j}{\bm v}_{\mathbf{Q},i}{\bm v}_{\mathbf{Q},i}^T
\Delta\mathbf{y}(j).
\end{aligned}
\end{equation}
where it should be noticed that according to Lemma~\ref{lema2} the last term vanishes to zero at $t$ approaches infinity.
Therefore, Eq.~\eqref{eq:last_derivation_convergence} becomes
\begin{equation}
\begin{aligned}
\lim_{t\rightarrow \infty }\mathbf{x}(t) &=\dfrac{1}{N}{\bm 1}\cdot{\bm 1}^T {\bm x}(0)+
\lim_{t\rightarrow \infty } \sum_{j=1}^{t} 
\Delta\overline{{y}}(j){\bm 1}\\
&=\overline{\mathbf{x}}(0){\bm 1}+
\sum_{j=1}^{\infty} 
\Delta\overline{{y}}(j){\bm 1}_n =x^*{\bm 1},
\end{aligned}
\end{equation}
thus concluding the proof.
\end{proof}

To conclude the section, we briefly analyze some of the properties and requirements of the proposed algorithm.
\begin{remark}[Communication demands]
Our algorithm requires the exchange of three values per agent and communication round. On one hand, each agent sends each time the value of \black{$c_i(t).$} This value is the same one used in the previous section to estimate $\bm\rho_\alpha$.
On the other hand, there are two sums in Eq.~\eqref{Eq:WeightedCons_Ind} over the set of neighbors of each agent, one that requires $\bar{c}_j(t)$ to compute the average of $\bm\rho_\alpha$ and the other one that requires $x_j(t)$ to compute the weighted consensus.
\end{remark}

\black{
\begin{remark} [Asymptotic 
 Stability] The average consensus estimator in Eq.~\eqref{Eq:WeightCons_Infl} can be seen as a discrete-time dynamic consensus~\cite{MZ-SM:10} with exogenous input, $\Delta\mathbf{y}(t),$   vanishing as proven in Theorem~\ref{Theorem2}.
 Using Corollary 3.1 in~\cite{MZ-SM:10} we can claim that our algorithm is asymptotically stable.
\end{remark}
}
\vspace{-8mm}

\section{Distributed Control of Node $\alpha$-Centrality}
\label{sec:control}

There are other applications involving networks, where rather than estimating centrality, we are interested in \emph{controlling} its value.  In order to do so, in this section we present an optimization problem that aims at performing minimum variations on the influence matrix to achieve this objective.
Interestingly, the proposed problem can be decomposed into local sub-problems that can be solved at each node via standard methods. Thus, it turns out to be very suitable for a distributed application context.

Let $\bm\rho_\alpha^*$ be the desired $\alpha$-centrality vector for a given graph $\mathcal{G}$ with initial influence matrix $\mathbf{W}$. 
We denote by $x_{ij}$ the amount of variation applied on the influence of agent $j$'s information for agent $i$, and by $\mathbf{X}=[x_{ij}]$ the matrix containing all the changes in the original influence matrix.
Our goal in this section is then to
find the matrix $\mathbf{X}^*\in \mathbb{A}_{\mathcal{G}}$ that solves
\begin{equation}
\label{prob:convexproblem0}
\begin{aligned}
& \underset{\mathbf{X}\in \mathbb{A}_{\mathcal{G}}}{\min}
& &\dfrac{1}{2}\|\mathbf{X}\|_F^2 \\
& \text{subject to}
& & \mathbf{W}+\mathbf{X} \leq \overline{\mathbf{W}} \\
& & &\mathbf{W}+\mathbf{X} \geq \underline{\mathbf{W}}\\
& & & (\mathbf{I}_N-\alpha(\mathbf{W} + \mathbf{X})^T)\bm\rho_\alpha^* = \mathbf{z}.
\end{aligned}\end{equation}
The optimization problem represents that we want minimum effort variation on the matrix $\mathbf{W}$; in this view, we are interested in scenarios where the nodes aim at  slightly modifying how their neighbors perceive their importance in the network while changing the $\alpha$-centrality to the desired value.
The first two constraints are included to model the fact that every influence value cannot change more than an arbitrary amount, encoded by the matrices $\overline{\mathbf{W}}$ and $\underline{\mathbf{W}}$.
The last constraint in Eq.~\eqref{prob:convexproblem0} imposes that the new influence matrix, $\mathbf{W}+\mathbf{X}$, has to yield the desired centrality value.
Notice that, by minimizing the Frobenius norm in the objective function, we explicitly consider a scenario where the variation at each link directly influences the objective function. 

Let ${\bf X}\in \mathbb{A}_G$ be partitioned as ${\bf X} =[{\bf X}_1 \ldots {\bf X}_N]$, where each ${\bf X}_i\in \mathbb{R}^N$. Similarly, let ${\bf W}$, $\overline {\bf W}$, and $\underline {\bf W}$ be partitioned as ${\bf W}=[{\bf W}_1,\ldots {\bf W}_N]$, $\overline {\bf W}=[\overline {\bf W}_1,\ldots \overline {\bf W}_N]$, $\underline {\bf W}=[\underline {\bf W}_1,\ldots \underline {\bf W}_N]$, where each ${\bf W}_i, \overline {\bf W}_i, \underline {\bf W}_i \in \mathbb{R}^N$. With such partitioning, the problem in Eq.~\eqref{prob:convexproblem0} can be equivalently expressed as a collection of $N$ local sub-problems in the form
 \begin{equation}
\label{prob:convexproblem00}
\begin{aligned}
& \underset{{\bf X}_i\in\mathbb{R}^N}{\min}
& &\dfrac{1}{2}{\bf X}_i^T {\bf X}_i\\
& \text{subject to}
& & \mathbf{W}_i+\mathbf{X}_i \leq \overline{\mathbf{W}}_i \\
& & &\mathbf{W}_i+\mathbf{X}_i \geq \underline{\mathbf{W}}_i\\
& & & \rho_{\alpha i}-\alpha(\mathbf{W}_i + \mathbf{X}_i)^T)\bm\rho_\alpha^* = {z}_i,
\end{aligned}
\end{equation} 
where $\underline{w}_{ij}=\overline{w}_{ij} =0$ whenever 
\black{$(i,j)\not \in \mathcal{E}$}, hence $x_{ij}$ is zero for \black{$(i,j)\not \in \mathcal{E}$}. 

\begin{remark}
The problem in Eq.~\eqref{prob:convexproblem00} is a quadratic programming problem with linear inequality constraints; hence, it can be solved using standard techniques/solvers. 
However, in order to solve the sub-problems locally, each node $i$ must know the coefficients $w_{ji},\overline w_{ji},\underline w_{ji},\rho^*_{\alpha j}$ associated to its neighbors; such an information can be obtained via a single communication round.
\end{remark}
\vspace{-5mm}

\subsection{Attack protection mechanisms}
As noted in early \cite{holme2002} and  more recent \cite{berezin2015localized} studies in complex network theory, attacks dealt to the nodes of a network (e.g., disrupting them) may have severe effects in terms of residual connectivity, especially when the attacker selects the targets based on topological features (e.g., degree, centrality, etc.).
Typical protection approaches (see, \cite{faramondi2018network} and references therein) are centralized and aim at prioritizing the protection of the most important nodes, with the aim to make all nodes equally valuable for the attacker.
In order to achieve this task in a distributed way, assuming that  the attacker's choices are driven by the $\alpha$-centrality vector, the control approach outlined in this section appears a valuable tool.
Specifically, in order to be protected against an attacker, the nodes may aim at hiding their true $\alpha$-centrality by forcing all values to be identical. To this end, it is reasonable to assume that the nodes want to modify the weights of the least possible amount (e.g., in oder to minimize the effort or to avoid that an attacker detects large changes).

\section{Simulations}
\label{sec:simulations}
\subsection{Centrality estimation}
In the first simulation we are going to show an example of the distributed estimation of the $\alpha$-centrality of a network and its application to weighted consensus with 15 nodes and topology shown in Fig.~\ref{Fig:Topology}.
We have numbered and assigned colors to each node to better highlight the centrality properties in the simulations.
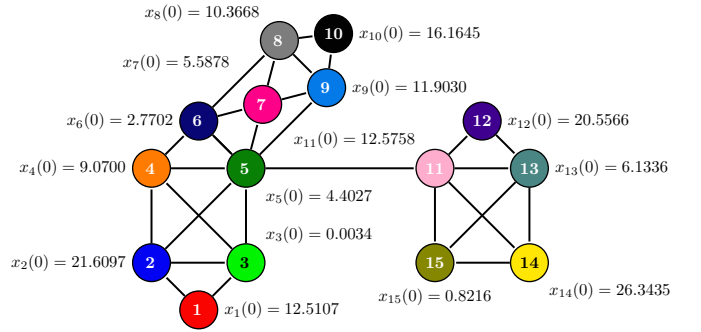
\begin{figure}[!ht]
    \centering
    
           \resizebox {.5\textwidth} {!}{
 \begin{tikzpicture}[
             > = stealth, 
             shorten > = 1pt, 
             auto,
             node distance = 16cm, 
             semithick 
         ]

         \tikzstyle{every state}=[
             draw = black,
             thick,
             fill = white,
             minimum size = 8mm,
             minimum width=2mm
         ]

  	\node[label={right:{$x_1(0)=12.5107$}},state,fill={rgb,255:red,249; green,4; blue,0}] (v1) at (1, -1) {\textcolor{white}{${\bm {1\,}}$}};
    \node[label={left:{$x_2(0)=21.6097$}},state,fill={rgb,255:red,2; green,3; blue,242}] (v2) at (0, 0) {\textcolor{white}{${\bm {2\,}}$}};
    \node[label={above right:{$x_3(0)=0.0034$}},state,fill={rgb,255:red,1; green,243; blue,1}] (v3) at (2,0) {\textcolor{black}{${\bm {3\,}}$}};
    \node[label={ left:{$x_4(0)=9.0700$}},state,fill={rgb,255:red,255; green,123; blue,2}] (v4) at (0, 2) {\textcolor{white}{${\bm {4\,}}$}};
    \node[label={below right:{$x_5(0)= 4.4027$}},state,fill={rgb,255:red,7; green,128; blue,5}] (v5) at (2, 2) {\textcolor{white}{${\bm {5\,}}$}};
    \node[label={ left:{$x_6(0)=2.7702$}},state,fill={rgb,255:red,8; green,6; blue,117}] (v6) at (1, 3) {\textcolor{white}{${\bm {6\,}}$}};
    \node[label={[label distance=0.4cm]above left:{$x_7(0)=5.5878$}},state,fill={rgb,255:red,254; green,1; blue,137}] (v7) at (2+0.7071/2, 3+0.7071/2) {\textcolor{white}{${\bm {7\,}}$}};
    \node[label={above left :{$x_8(0)=10.3668$}},state,fill={rgb,255:red,125; green,125; blue,125}] (v8) at (2+0.7071, 4+0.7071)  {\textcolor{white}{${\bm {8\,}}$}};
    \node[label={right:{$x_9(0)=11.9030$}},state,fill={rgb,255:red,4; green,122; blue,233}] (v9) at (3+0.7071, 3+0.7071)  {\textcolor{white}{${\bm {9\,}}$}};
    \node[label={right:{$x_{10}(0)=16.1645$}},state,fill={rgb,255:red,2; green,2; blue,2}] (v10) at (1.5+2+0.7071/2, 1.5+3+0.7071/2)  {\textcolor{white}{${\bm {10}}$}};
    \node[label={above left:{$x_{11}(0)=12.5758$}},state,fill={rgb,255:red,255; green,173; blue,205}] (v11) at (4+2, 2)  {\textcolor{white}{${\bm {11}}$}};
    \node[label={right:{$x_{12}(0)=20.5566$}},state,fill={rgb,255:red,64; green,0; blue,143}] (v12) at (5+2, 3)  {\textcolor{white}{${\bm {12}}$}};
    \node[label={right:{$x_{13}(0)=6.1336$}},state,fill={rgb,255:red,73; green,134; blue,132}] (v13) at (6+2, 2)  {\textcolor{white}{${\bm {13}}$}};
    \node[label={below right:{$x_{14}(0)=26.3435$}},state,fill={rgb,255:red,255; green,230; blue,5}] (v14) at (6+2, 0) {\textcolor{black}{${\bm {14}}$}};
    \node[label={below:{$x_{15}(0)=0.8216$}},state,fill={rgb,255:red,134; green,135; blue,0}] (v15) at (4+2, 0) {\textcolor{white}{${\bm {15}}$}};

\path[sloped,-, above,very thick] (v1) edge node {} (v2);
\path[sloped,-, above,very thick] (v1) edge node {} (v3);
\path[sloped,-, above,very thick] (v2) edge node {} (v3);
\path[sloped,-, above,very thick] (v2) edge node {} (v4);
\path[sloped,-, above,very thick] (v2) edge node {} (v5);
\path[sloped,-, above,very thick] (v3) edge node {} (v4);
\path[sloped,-, above,very thick] (v3) edge node {} (v5);
\path[sloped,-, above,very thick] (v4) edge node {} (v5);
\path[sloped,-, above,very thick] (v4) edge node {} (v6);
\path[sloped,-, above,very thick] (v5) edge node {} (v6);
\path[sloped,-, above,very thick] (v5) edge node {} (v7);
\path[sloped,-, above,very thick] (v5) edge node {} (v9);
\path[sloped,-, above,very thick] (v5) edge node {} (v11);
\path[sloped,-, above,very thick] (v6) edge node {} (v7);
\path[sloped,-, above,very thick] (v6) edge node {} (v8);
\path[sloped,-, above,very thick] (v6) edge node {} (v5);
\path[sloped,-, above,very thick] (v7) edge node {} (v8);
\path[sloped,-, above,very thick] (v7) edge node {} (v9);
\path[sloped,-, above,very thick] (v8) edge node {} (v9);
\path[sloped,-, above,very thick] (v8) edge node {} (v10);
\path[sloped,-, above,very thick] (v9) edge node {} (v10);
\path[sloped,-, above,very thick] (v11) edge node {} (v12);
\path[sloped,-, above,very thick] (v11) edge node {} (v13);
\path[sloped,-, above,very thick] (v11) edge node {} (v14);
\path[sloped,-, above,very thick] (v11) edge node {} (v15);
\path[sloped,-, above,very thick] (v12) edge node {} (v13);
\path[sloped,-, above,very thick] (v13) edge node {} (v14);
\path[sloped,-, above,very thick] (v13) edge node {} (v15);
\path[sloped,-, above,very thick] (v14) edge node {} (v15);
     \end{tikzpicture}
 }  

    \caption{Network topology. The initial condition $x_i(0)$ for the average consensus is reported next to each node.}
    \label{Fig:Topology}
\end{figure}

First, we consider the distributed computation of $\bm\rho_\alpha$ associated to the adjacency matrix, i.e., $\mathbf{W}=\mathbf{A}$ and symmetric, with a uniform initial importance vector, $\mathbf{z}=\mathbf{1}.$
This way, there is a direct relationship between the centrality value and the connectivity of each node, resulting in node $5$, in green, having the highest centrality value and nodes $1$, ${10}$ and ${12}$ having the lowest values.
The evolution of $\mathbf{c}(t)$ is depicted in Fig.~\ref{Fig:Centrality} (a).
The parameter $\alpha$ has been chosen in the simulation equal to $0.8/\|\mathbf{W}\|$.
Since the matrix is symmetric, the bound in Theorem~\ref{Theorem1} reduces to $\gamma=1$ and $\|\cdot\|_{\mathbf{W}}$ 
the spectral norm, reducing the error by a factor of $0.8$ at each communication round. Considering that in this particular case $\|\mathbf{z}\|/(1-\kappa)
\simeq 19,$ our analytic bound states that the algorithm should reach an accuracy below $0.1$ in approximately 24 communication rounds which is consistent with the plot. 
The difference between the actual estimation error, $e_c(t)$ and the theoretical bound in Eq.~\eqref{Eq:errorBound} is shown in Fig.~\ref{Fig:Centrality} (b), where we can see that this difference is not only positive for all $t$, but also close to zero. 

\begin{figure}[!ht]
\begin{center}
    \subfloat[]{\label{fig1:convrate} \includegraphics[height=0.44\columnwidth]{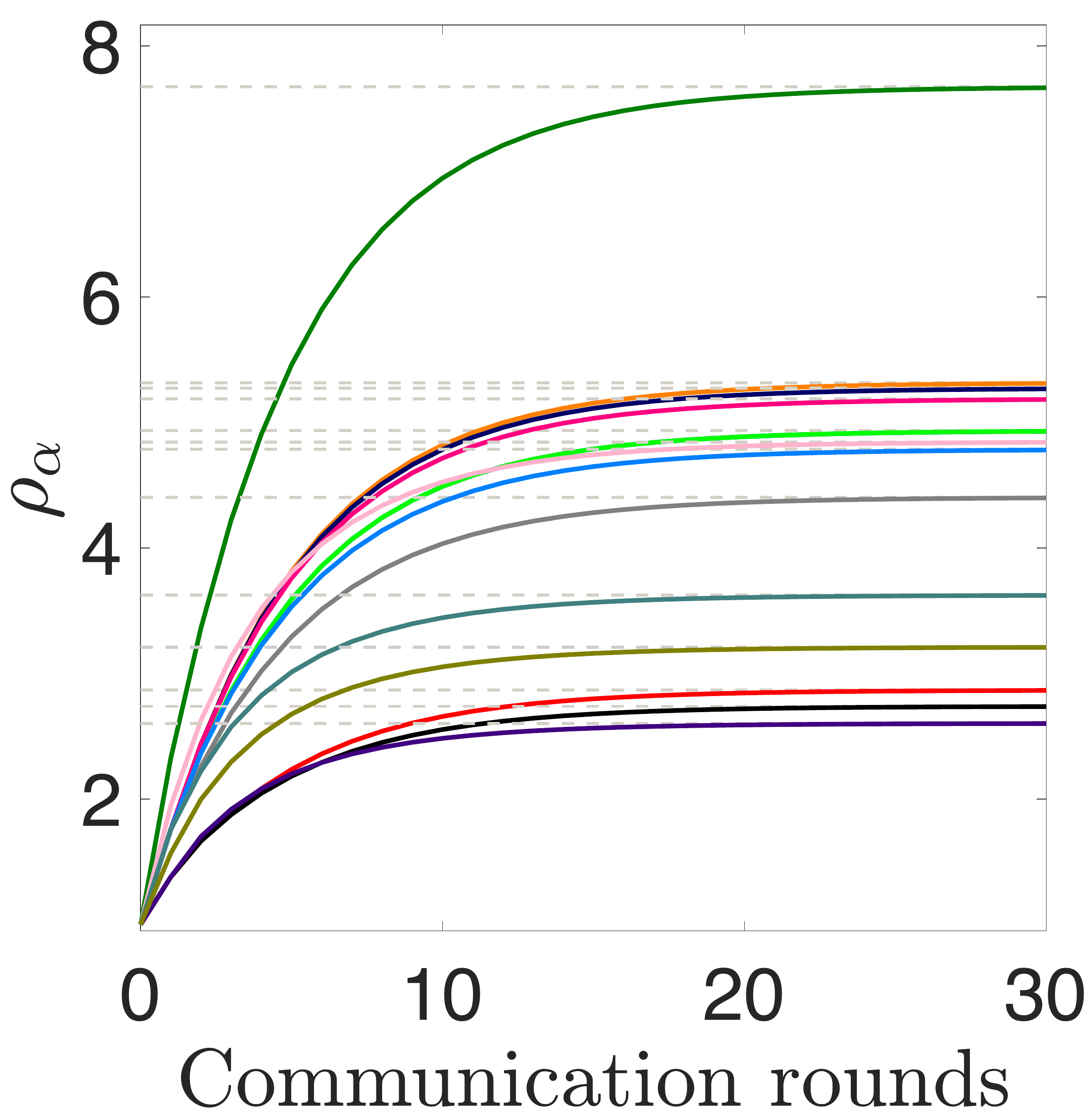}}
    \hspace{1mm}
    \subfloat[]{\label{fig1:convrate} 
    \includegraphics[height=0.449\columnwidth]{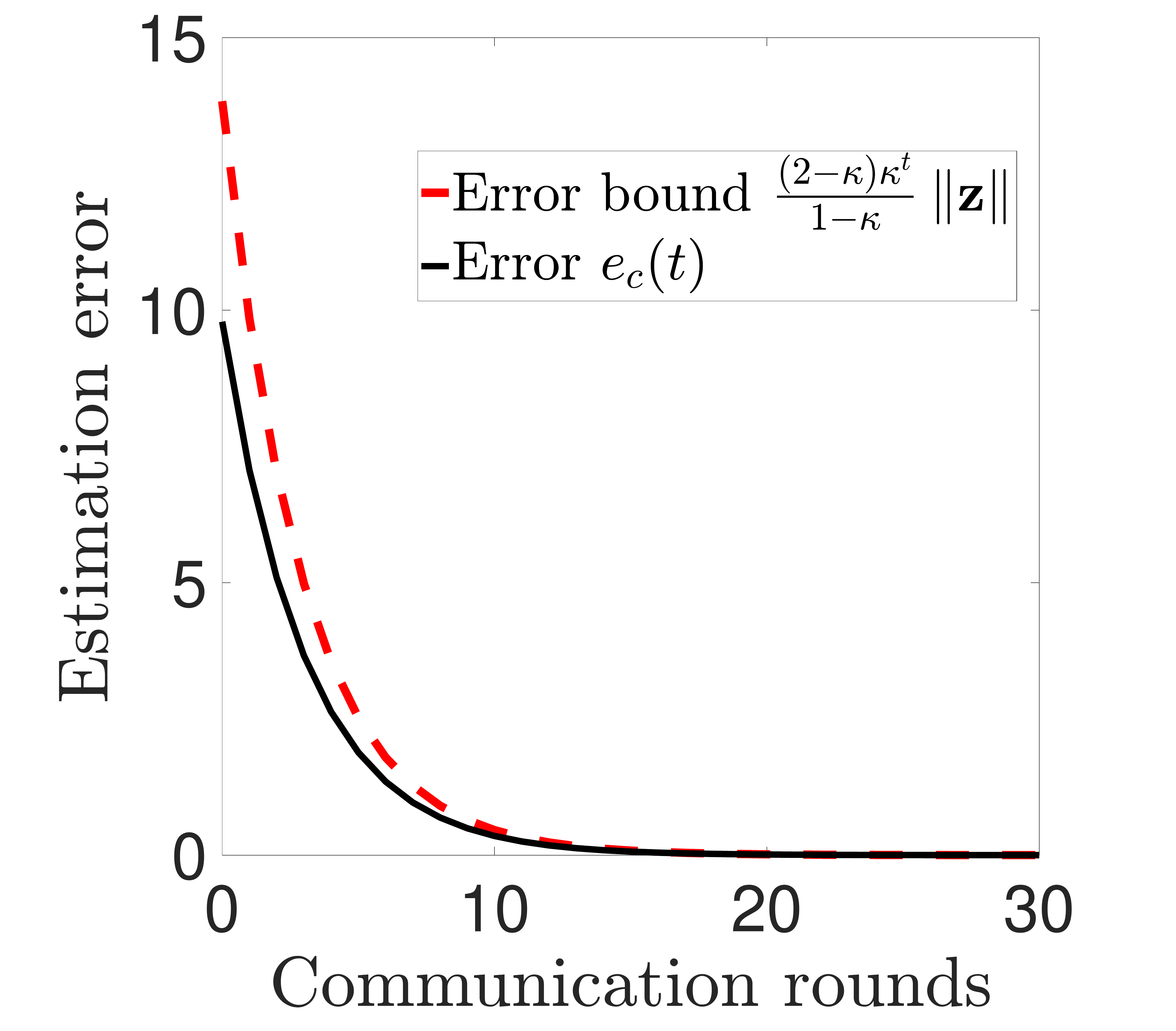}}
\end{center}
    \caption{Distributed estimation of $\bm\rho_\alpha$ for the adjacency matrix and $\mathbf{z}=\mathbf{1}$. Panel~(a): Estimation of $\bm\rho_\alpha$. Panel~(b): Estimation error.
    }
    \label{Fig:Centrality}
\end{figure}

\begin{figure}[!ht]
\begin{center}
    \subfloat[]{\label{fig1:convrate} \includegraphics[height=0.45\columnwidth]{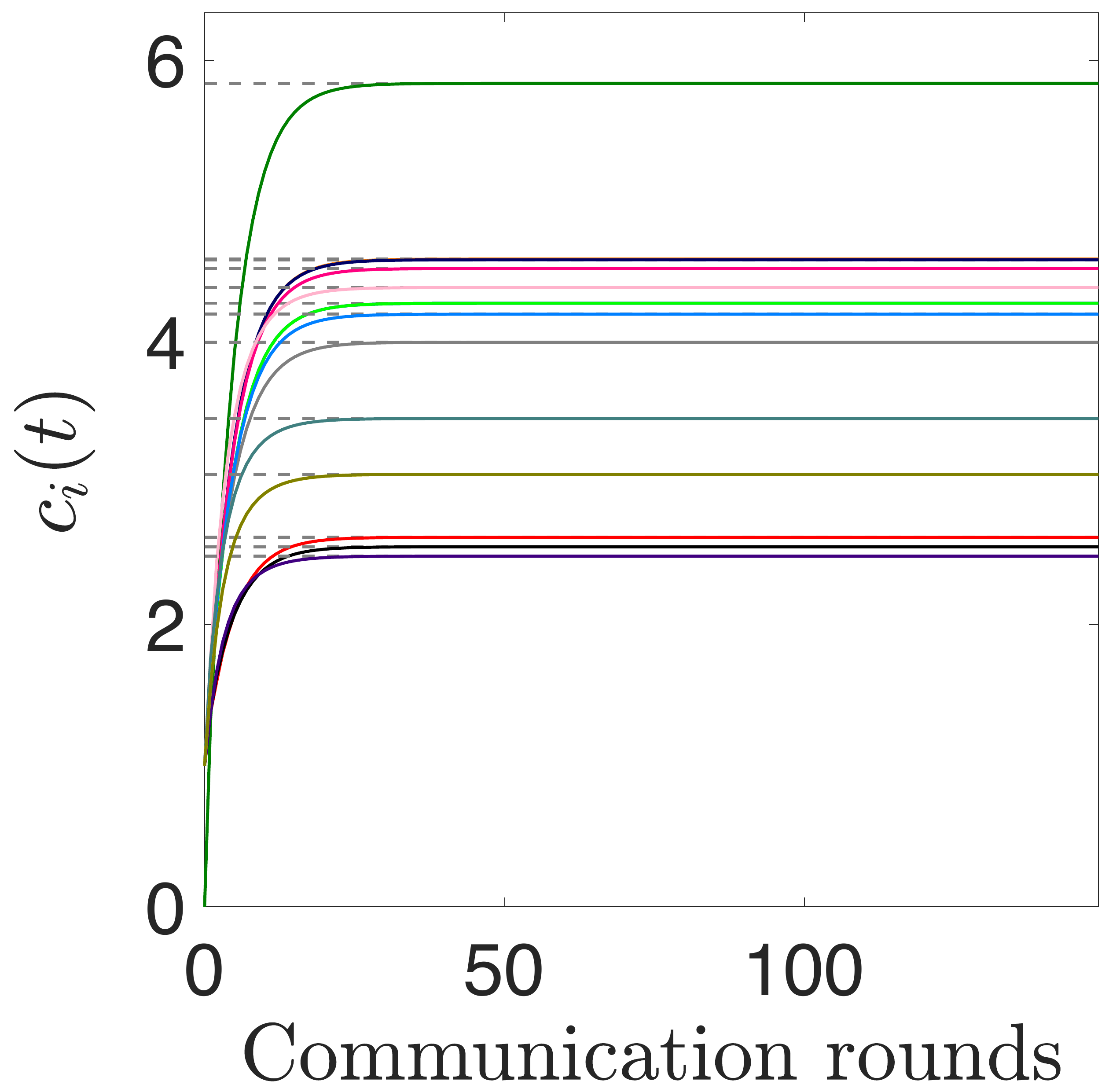}}
    \subfloat[]{\label{fig1:convrate} \includegraphics[height=0.45\columnwidth]{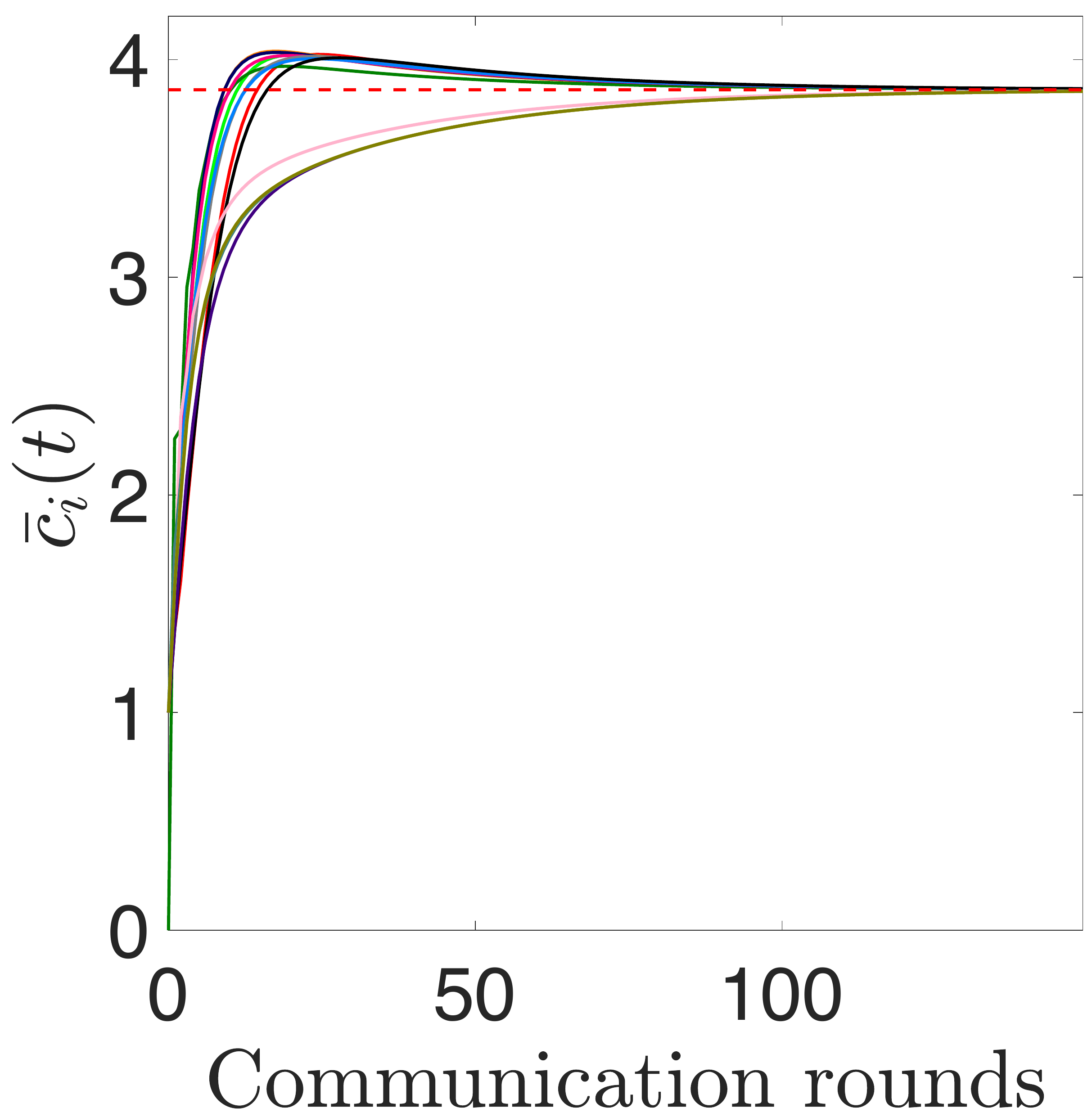}}
    \vspace{2mm}
        \subfloat[]{\label{fig1:convrate} \includegraphics[height=0.45\columnwidth]{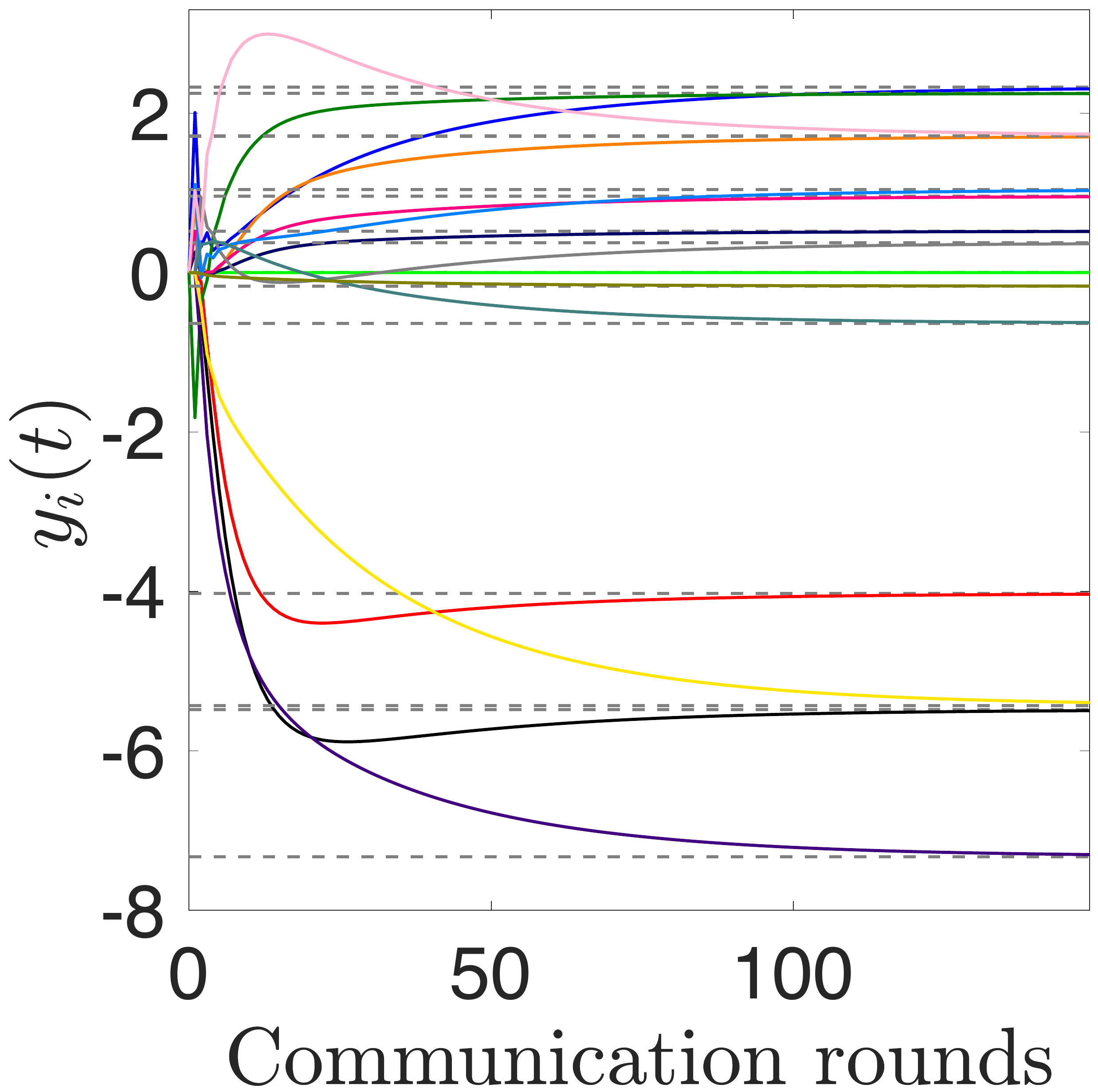}}
    \subfloat[]{\label{fig1:convrate} \includegraphics[height=0.459\columnwidth]{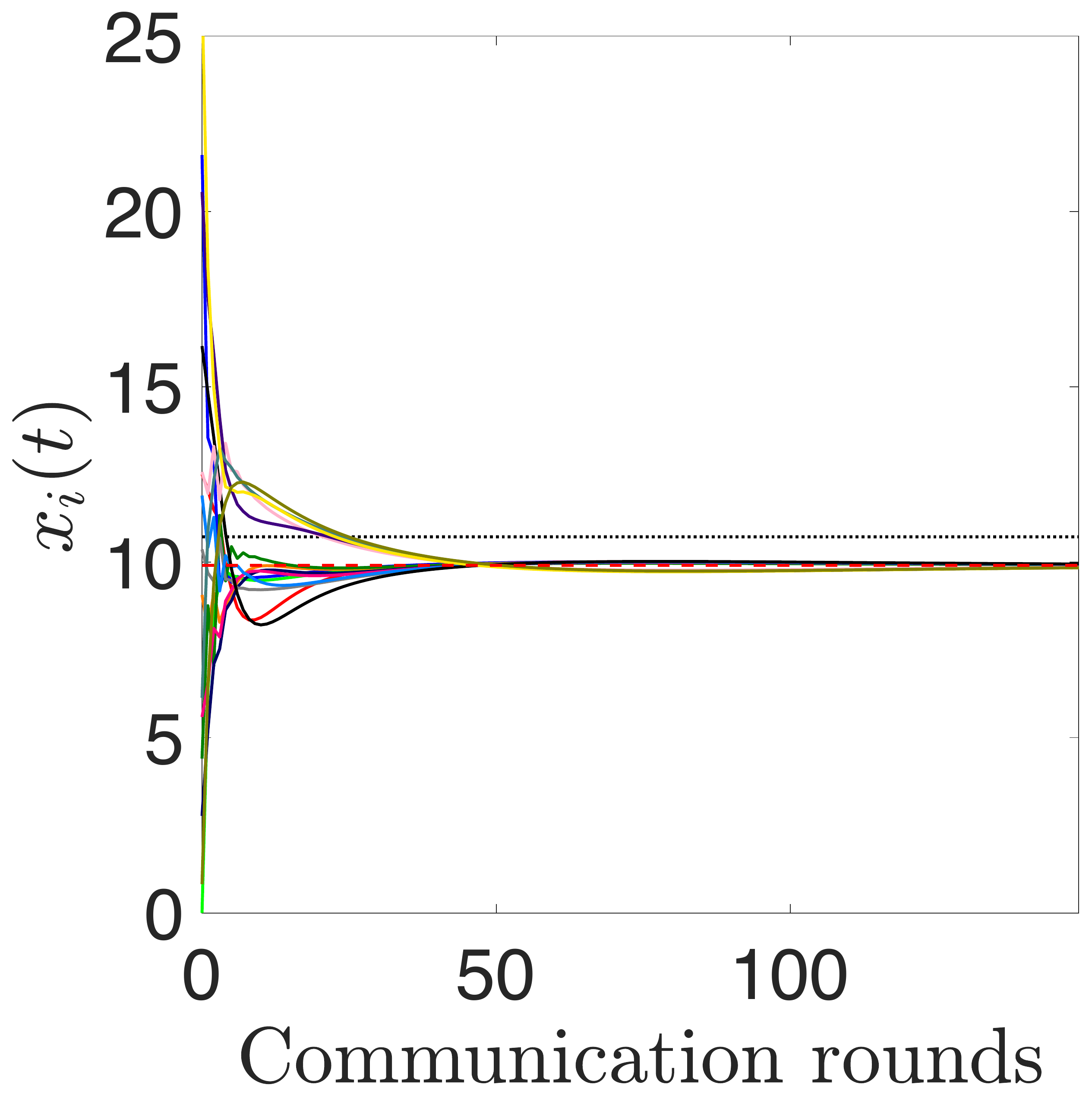}}
\end{center}
    \caption{Influence-based consensus with influence matrix equal to the adjacency matrix and $z_5=0$.
    Panels~(a)--(d) show the evolution of $\mathbf{c}(t)$, $\bar{\mathbf{c}}(t)$, $\mathbf{y}(t)$, and $\mathbf{x}(t)$, respectively.
    }
    \label{Fig:InfConsensus}
\end{figure}

Successively,
we combine the centrality estimation method together with the influence-based consensus proposed in Section~\ref{sec:influence}, considering the initial conditions ${\bm x}(0)$ shown in Figure~\ref{Fig:Topology}.
To highlight the practical implications of Lemma~\ref{lema3}
we consider now an initial importance vector such that $z_5=0$ and $z_i=1$ for every other node $i\neq 5$. 
In Figure~\ref{Fig:InfConsensus} we show the evolution of the four variables analyzed in Theorem~\ref{Theorem2} that do not represent increments, i.e., $\mathbf{c}(t), \bar{\mathbf{c}}(t), \mathbf{y}(t)$ and $\mathbf{x}(t)$.

The top left plot in Figure~\ref{Fig:InfConsensus} shows the new estimation of the centrality vector.
The difference in the initial importance vector leads to a different final centralities. Setting $z_5=0$ we observe a decrease in the final centrality value of node 5, from slightly less than $8$ in Fig~\ref{Fig:Centrality} (a) to slightly less than $6$ in the new simulation.
In the top right plot in Figure~\ref{Fig:InfConsensus} we can observe how all the nodes in the network reach asymptotically the average of $\bm\rho_\alpha$, shown as grey dashed line.
The convergence speed is slower than for the computation of $\mathbf{c}(t)$ due to the slower convergence rate of the powers of $\mathbf{Q}$.

The bottom left plot in Figure~\ref{Fig:InfConsensus} shows the convergence of $\mathbf{y}(t)$ to the desired exogenous input in Eq.~\eqref{Eq:correctionTerm}. 
Note how this input is positive for the most influential nodes, like node $5$ (green) in Fig.~\ref{Fig:Topology} (a), whereas is negative for the least influential nodes, like node ${12}$ (purple).
This is consistent with the idea of giving more weight to the values of the most influential nodes, which in our setup is transformed into increasing their initial condition for the consensus algorithm.
%
Finally, the bottom right plot in Figure~\ref{Fig:InfConsensus} shows the consensus evolution of the initial conditions. For the sake of visualization, we have also included in the plot the value of the average (black dotted line), to better visualize that our algorithm does not converge to this value but to the weighted average (red dashed line) in Eq.~\eqref{Eq:weightedConsensus}.

\subsection{Centrality control}
We provide an example of application of the proposed centrality control scheme. Specifically, we consider the network reported in Figure~\ref{fig:example:1}, for which we have $\alpha$-centrality
${\bm \rho}_\alpha=[8,7,6,5,4,3]^T$.
For security reasons, the network in Figure~\ref{fig:example:1} needs to reach a configuration where all nodes have equal $\alpha$-centrality, and specifically $\alpha$-centrality equal to ${\bm 1}_n$, by modifying the original weights of the least possible amount (in a least square sense).
Let us consider an initial value $\mathbf{z}$ that is proportional to the $\alpha$-centrality given above and, specifically, 
$
\mathbf{z}={\bm \rho}_\alpha/({\bm \rho}_\alpha^T\bm {1}_N).
$
Moreover, let us choose $\underline{w}_{ij}=\underline w=1.5$ and $\overline{w}_{ij}=\overline w=5$, for all edges. 
In Figure~\ref{fig:example:2} we report along the edges the values $x^*_{ij}+w_{ij}$ obtained via a standard quadratic programming solver\footnote{For simplicity we use the quadprog solver in Matlab\texttrademark.}.
Notice that we obtain \mbox{$\|{\bm X}^*\|^2_F/2=4.6742$}, i.e., we are able to make all nodes equal in terms of $\alpha$-centrality with a small variation of the weights.
\begin{remark}
In this example we choose ${\bm \rho}^*_{\alpha}={\bm 1}_N$.
An interesting extension for future work would be to set \mbox{${\bm \rho}^*_{\alpha}=\rho{\bm 1}_N$} and let the agents collectively chose the value $\rho$ that corresponds to the minimum variation in the influence matrix.
\end{remark}
 \color{black}
 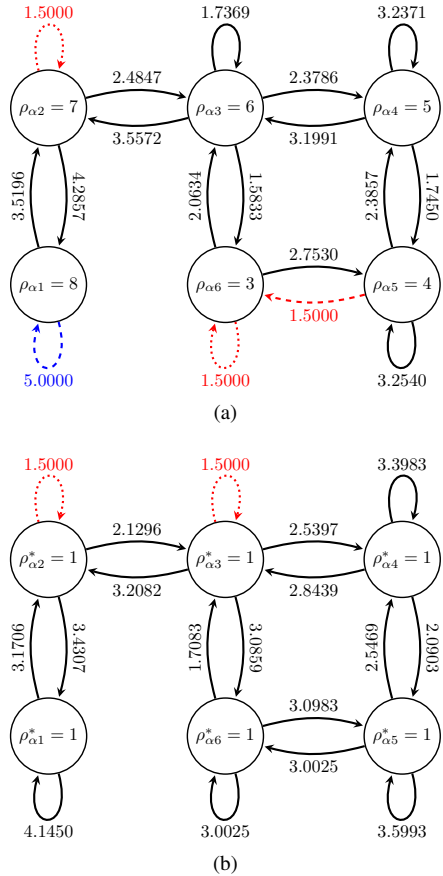
\begin{figure}
 \begin{center} 
 \begin{tabular}{c}
 \subfloat[]{
 \label{fig:example:1}
 \resizebox {.33\textwidth} {!}{
 \begin{tikzpicture}[
             > = stealth, 
             shorten > = 1pt, 
             auto,
             node distance = 3.5cm, 
             semithick 
         ]

         \tikzstyle{every state}=[
             draw = black,
             thick,
             fill = white,
             minimum size = 4mm,
             minimum width=15mm
         ]

  	\node[state] (v1) {$\rho_{\alpha 1}=8$};
         \node[state] (v2) [above of=v1] {$\rho_{\alpha 2}=7$};
         \node[state] (v3) [ right of=v2] {$\rho_{\alpha 3}=6$};
         \node[state] (v4) [right of=v3] {$\rho_{\alpha 4}=5$};
         \node[state] (v5) [below of=v4] {$\rho_{\alpha 5}=4$};
         \node[state] (v6) [left of=v5] {$\rho_{\alpha 6}=3$};

 \path[->,loop below,very thick,dashed,blue] (v1) edge node {$5.0000$} (v1);
 \path[->,loop above,very thick,dotted, red] (v2) edge node {$1.5000$} (v2);
 \path[->,loop above,very thick] (v3) edge node {$1.7369$} (v3);
 \path[->,loop above,very thick] (v4) edge node {$3.2371$} (v4);
 \path[->,loop below,very thick] (v5) edge node {$3.2540$} (v5);
 \path[->,loop below,very thick,dotted, red] (v6) edge node {$1.5000$} (v6);
             
 \path[sloped,->, above, bend left=15,very thick] (v1) edge node {$3.5196$} (v2);
 \path[sloped,->, above, bend left=15,very thick] (v2) edge node {$4.2857$} (v1);

 \path[sloped,->, above, bend left=15,very thick] (v2) edge node {$2.4847$} (v3);
 \path[sloped,->, below, bend left=15,very thick] (v3) edge node {$3.5572$} (v2);
 \path[sloped,->, above, bend left=15,very thick] (v3) edge node {$2.3786$} (v4);
 \path[sloped,->, below, bend left=15,very thick] (v4) edge node {$3.1991$} (v3);

 \path[sloped,->, above, bend left=15,very thick] (v4) edge node {$1.7450$} (v5);
 \path[sloped,->, above, bend left=15,very thick] (v5) edge node {$2.3857 $} (v4);

 \path[sloped,->, below, bend left=15,very thick,dashed, red] (v5) edge node {$1.5000$} (v6);
 \path[sloped,->, above, bend left=15,very thick] (v6) edge node {$2.7530 $} (v5);

 \path[sloped,->, above, bend left=15,very thick] (v3) edge node {$1.5833$} (v6);
 \path[sloped,->, above, bend left=15,very thick] (v6) edge node {$2.0634$} (v3);
     \end{tikzpicture}
 }
 }
\\
 \subfloat[]{
 \label{fig:example:2}
 \resizebox {.33\textwidth} {!}{
 \begin{tikzpicture}[
             > = stealth, 
             shorten > = 1pt, 
             auto,
             node distance = 3.5cm, 
             semithick 
         ]

         \tikzstyle{every state}=[
             draw = black,
             thick,
             fill = white,
             minimum size = 4mm,
             minimum width=15mm
         ]

  	\node[state] (v1) {$\rho^*_{\alpha 1}=1$};
         \node[state] (v2) [above of=v1] {$\rho^*_{\alpha 2}=1$};
         \node[state] (v3) [ right of=v2] {$\rho^*_{\alpha 3}=1$};
         \node[state] (v4) [right of=v3] {$\rho^*_{\alpha 4}=1$};
         \node[state] (v5) [below of=v4] {$\rho^*_{\alpha 5}=1$};
         \node[state] (v6) [left of=v5] {$\rho^*_{\alpha 6}=1$};

 \path[->,loop below,very thick] (v1) edge node {$4.1450 $} (v1);
 \path[->,loop above,very thick,dotted, red] (v2) edge node {$1.5000$} (v2);
 \path[->,loop above,very thick,dotted, red] (v3) edge node {$1.5000$} (v3);
 \path[->,loop above,very thick] (v4) edge node {$3.3983$} (v4);
 \path[->,loop below,very thick] (v5) edge node {$3.5993$} (v5);
 \path[->,loop below,very thick] (v6) edge node {$3.0025$} (v6);
             
 \path[sloped,->, above, bend left=15,very thick] (v1) edge node {$3.1706$} (v2);
 \path[sloped,->, above, bend left=15,very thick] (v2) edge node {$3.4307$} (v1);

 \path[sloped,->, above, bend left=15,very thick] (v2) edge node {$2.1296$} (v3);
 \path[sloped,->, below, bend left=15,very thick] (v3) edge node {$3.2082$} (v2);
 \path[sloped,->, above, bend left=15,very thick] (v3) edge node {$2.5397$} (v4);
 \path[sloped,->, below, bend left=15,very thick] (v4) edge node {$2.8439$} (v3);

 \path[sloped,->, above, bend left=15,very thick] (v4) edge node {$2.0903$} (v5);
 \path[sloped,->, above, bend left=15,very thick] (v5) edge node {$2.5469$} (v4);

 \path[sloped,->, below, bend left=15,very thick] (v5) edge node {$3.0025$} (v6);
 \path[sloped,->, above, bend left=15,very thick] (v6) edge node {$3.0983$} (v5);

 \path[sloped,->, above, bend left=15,very thick] (v3) edge node {$3.0859$} (v6);
 \path[sloped,->, above, bend left=15,very thick] (v6) edge node {$1.7083$} (v3);
     \end{tikzpicture}
 }
 }
  \end{tabular}
 \caption{Example of $\alpha$-centrality control for a network with $n=6$ nodes. Panel (a) reports the initial coefficients $w_{ij}$ along each edge and the $\alpha$-centrality value at each node $v_i$. Panel (b) shows the resulting weights $w_{ij}+x^*_{ij}$ and  $\alpha$-centrality values as a result of the local solution of the sub-problem in Eq.~\eqref{prob:convexproblem00} by each node.
 We show in red dotted lines and blue dashed lines the links where the weights reach the lower and upper bound, respectively.
 }
  \label{fig:example}
 \end{center}
 \end{figure}

\section{Conclusions}
\label{sec:conclusions}
In this work, the problems  of distributed node centrality identification and control have been addressed. We have developed a protocol for the distributed computation of $\alpha$-centrality, 
which is particularly suitable for networks with asymmetric interactions. 
We have also discussed a local solution for the computation of minimum variation of weights such that the network yields a desired centrality value.
In addition, motivated by studies on social networks, we have proposed a novel consensus-based
algorithm which runs in parallel to the $\alpha$-centrality estimation and achieves a weighted consensus, where
the weights are given precisely by the values of the $\alpha$-centrality. 
The control algorithm has also been applied to the problem of minimizing agents' vulnerability to external influences.

\bibliographystyle{IEEEtran}
\bibliography{references}%

\begin{thebibliography}{10}
\providecommand{\url}[1]{#1}
\csname url@samestyle\endcsname
\providecommand{\newblock}{\relax}
\providecommand{\bibinfo}[2]{#2}
\providecommand{\BIBentrySTDinterwordspacing}{\spaceskip=0pt\relax}
\providecommand{\BIBentryALTinterwordstretchfactor}{4}
\providecommand{\BIBentryALTinterwordspacing}{\spaceskip=\fontdimen2\font plus
\BIBentryALTinterwordstretchfactor\fontdimen3\font minus
  \fontdimen4\font\relax}
\providecommand{\BIBforeignlanguage}[2]{{%
\expandafter\ifx\csname l@#1\endcsname\relax
\typeout{** WARNING: IEEEtran.bst: No hyphenation pattern has been}%
\typeout{** loaded for the language `#1'. Using the pattern for}%
\typeout{** the default language instead.}%
\else
\language=\csname l@#1\endcsname
\fi
#2}}
\providecommand{\BIBdecl}{\relax}
\BIBdecl

\bibitem{Freeman:1978}
L.~C. Freeman, ``Centrality in social networks conceptual clarification,''
  \emph{Social Networks}, vol.~1, no.~3, pp. 215 -- 239, 1978.

\bibitem{Newman:2010}
M.~Newman, \emph{Networks: An Introduction}.\hskip 1em plus 0.5em minus
  0.4em\relax New York, NY, USA: Oxford University Press, Inc., 2010.

\bibitem{wasserman1994social}
S.~Wasserman and K.~Faust, \emph{Social network analysis: Methods and
  applications}.\hskip 1em plus 0.5em minus 0.4em\relax Cambridge university
  press, 1994, vol.~8.

\bibitem{oliva2016distributed}
G.~Oliva, R.~Setola, and C.~N. Hadjicostis, ``Distributed finite-time
  calculation of node eccentricities, graph radius and graph diameter,''
  \emph{Systems \& Control Letters}, vol.~92, pp. 20--27, 2016.

\bibitem{Bavelas:1950}
A.~Bavelas, ``{Communication Patterns in Task-Oriented Groups},'' \emph{The
  Journal of the Acoustical Society of America}, vol.~22, no.~6, pp. 725--730,
  Nov. 1950.

\bibitem{Freeman:1977}
L.~C. Freeman, ``{A Set of Measures of Centrality Based on Betweenness},''
  \emph{Sociometry}, vol.~40, no.~1, pp. 35--41, Mar. 1977.

\bibitem{Bonacich:1972}
P.~Bonacich, ``Factoring and weighting approaches to status scores and clique
  identification,'' \emph{The Journal of Mathematical Sociology}, vol.~2,
  no.~1, pp. 113--120, 1972.

\bibitem{Bonacich:2001}
P.~Bonacich and P.~Lloyd, ``Eigenvector-like measures of centrality for
  asymmetric relations,'' \emph{Social Networks}, vol.~23, no.~3, pp. 191 --
  201, 2001.

\bibitem{Lehmann:2003}
K.~Anna~Lehmann and M.~Kaufmann, ``Decentralized algorithms for evaluating
  centrality in complex networks,'' Wilhelm-Schickard-Institut, Tech. Rep.
  WSI-2003-10, October 2003.

\bibitem{Wehmuth:2012}
K.~Wehmuth and A.~Ziviani, ``Distributed assessment of the closeness centrality
  ranking in complex networks,'' in \emph{Fourth Annual Workshop on Simplifying
  Complex Networks for Practitioners}, 2012, pp. 43--48.

\bibitem{Tang:2013}
W.~Wang and C.~Y. Tang, ``Distributed computation of node and edge betweenness
  on tree graphs,'' in \emph{52nd IEEE Conf. on Decision and Control}, Dec
  2013, pp. 43--48.

\bibitem{Tang:2014}
------, ``Distributed computation of classic and exponential closeness on tree
  graphs,'' in \emph{2014 American Control Conf.}, June 2014.

\bibitem{Tang:2015}
------, ``Distributed estimation of closeness centrality,'' in \emph{54th IEEE
  Conference on Decision and Control}, Dec 2015, pp. 4860--4865.

\bibitem{maccari2018distributed}
L.~Maccari, L.~Ghiro, A.~Guerrieri, A.~Montresor, and R.~L. Cigno, ``On the
  distributed computation of load centrality and its application to dv
  routing,'' in \emph{IEEE Int. Conf. on Computer Communications}, 2018.

\bibitem{Rossi:18}
W.~S. Rossi and P.~Frasca, ``On the convergence of message passing computation
  of harmonic influence in social networks,'' \emph{IEEE Transactions on
  Network Science and Engineering}, pp. 1--1, 2018, to appear.

\bibitem{Rossi:16}
------, ``An index for the local influence in social networks,'' in
  \emph{European Control Conference}, June 2016, pp. 525--530.

\bibitem{Ishii:14}
H.~Ishii and R.~Tempo, ``The pagerank problem, multiagent consensus, and web
  aggregation: A systems and control viewpoint,'' \emph{IEEE Control Systems
  Magazine}, vol.~34, no.~3, pp. 34--53, June 2014.

\bibitem{Charalambous:2016:CDC}
T.~Charalambous, C.~N. Hadjicostis, M.~G. Rabbat, and M.~Johansson, ``Totally
  asynchronous distributed estimation of eigenvector centrality in digraphs
  with application to the pagerank problem,'' in \emph{55th IEEE Conf. on
  Decision and Control}, Dec 2016, pp. 25--30.

\bibitem{Tempo:2017}
K.~You, R.~Tempo, and L.~Qiu, ``Distributed algorithms for computation of
  centrality measures in complex networks,'' \emph{IEEE Trans. on Automatic
  Control}, vol.~62, no.~5, pp. 2080--2094, May 2017.

\bibitem{Suzuki:18}
A.~Suzuki and H.~Ishii, ``Distributed randomized algorithms for pagerank based
  on a novel interpretation,'' in \emph{American Control Conference}, June
  2018, pp. 472--477.

\bibitem{Granovetter:1973}
M.~S. Granovetter, ``The strength of weak ties,'' \emph{American Journal of
  Sociology}, vol.~78, no.~6, pp. 1360--1380, 1973.

\bibitem{LusseauS477}
D.~Lusseau and M.~E.~J. Newman, ``Identifying the role that animals play in
  their social networks,'' \emph{Proc. of the Royal Society of London B:
  Biological Sciences}, vol. 271, no. Suppl 6, pp. S477--S481, 2004.

\bibitem{EM-GO-AG:18}
E.~{Montijano}, G.~{Oliva}, and A.~{Gasparri}, ``Distributed estimation of node
  centrality with application to agreement problems in social networks,'' in
  \emph{2018 IEEE Conference on Decision and Control (CDC)}, Dec 2018, pp.
  5245--5250.

\bibitem{Stoer92}
J.~Stoer and R.~Bulirsch, \emph{Introduction to numerical analysis}.\hskip 1em
  plus 0.5em minus 0.4em\relax Springer, 1992.

\bibitem{lynch1996distributed}
N.~A. Lynch, \emph{Distributed algorithms}.\hskip 1em plus 0.5em minus
  0.4em\relax Elsevier, 1996.

\bibitem{Sundaram:2008}
S.~Sundaram and C.~N. Hadjicostis, ``Distributed function calculation and
  consensus using linear iterative strategies,'' \emph{IEEE Journal on Selected
  Areas in Communications}, vol.~26, no.~4, pp. 650--660, 2008.

\bibitem{FB-JC-SM:09}
F.~Bullo, J.~Cort\'es, and S.~Mart{\'\i}nez, \emph{Distributed Control of
  Robotic Networks}, ser. Applied Mathematics Series.\hskip 1em plus 0.5em
  minus 0.4em\relax Princeton University Press, 2009, electronically available
  at http://coordinationbook.info.

\bibitem{Priolo:2014}
A.~Priolo, A.~Gasparri, E.~Montijano, and C.~Sagues, ``A distributed algorithm
  for average consensus on strongly connected weighted digraphs,''
  \emph{Automatica}, vol.~50, no.~3, pp. 946--951, 2014.

\bibitem{Montijano:2014}
E.~Montijano, J.~I. Montijano, C.~Sagues, and S.~Mart\'inez, ``Robust discrete
  time dynamic average consensus,'' \emph{Automatica}, vol.~50, no.~12, pp.
  3131--3138, 2014.

\bibitem{MZ-SM:10}
M.~Zhu and S.~Mart{\'\i}nez, ``{Discrete-time Dynamic Average Consensus},''
  \emph{Automatica}, vol.~46, no.~2, pp. 322--329, February 2010.

\bibitem{holme2002}
P.~Holme, B.~J. Kim, C.~N. Yoon, and S.~K. Han, ``Attack vulnerability of
  complex networks,'' \emph{Physical Review E}, vol.~65, no.~5, 2002.

\bibitem{berezin2015localized}
Y.~Berezin, A.~Bashan, M.~M. Danziger, D.~Li, and S.~Havlin, ``Localized
  attacks on spatially embedded networks with dependencies,'' \emph{Scientific
  reports}, vol.~5, 2015.

\bibitem{faramondi2018network}
L.~Faramondi, G.~Oliva, S.~Panzieri, F.~Pascucci, M.~Schlueter, M.~Munetomo,
  and R.~Setola, ``Network structural vulnerability: A multiobjective attacker
  perspective,'' \emph{IEEE Transactions on Systems, Man, and Cybernetics:
  Systems}, no.~99, pp. 1--14, 2018.

\end{thebibliography}

\end{document}